\documentclass[format=acmsmall, review=false]{acmart}
\pdfoutput=1
\usepackage{acm-ec-18}
\usepackage{natbib,booktabs}
\usepackage{wrapfig,subcaption}
\usepackage{enumerate,setspace}
\usepackage{amsmath,amssymb}

\usepackage[english, vlined, ruled, linesnumbered]{algorithm2e}
\usepackage{enumerate,setspace}
\usepackage{graphicx}
\usepackage{tikz}

\usetikzlibrary{decorations.pathreplacing, decorations.pathmorphing, positioning, calc}
\usetikzlibrary{arrows,shapes,  fit, positioning}
\usepackage{tkz-graph}

\newcommand{\calN}{\mathcal{N}}

\newcommand{\remain}{\mathrm{left}}
\newcommand{\single}{\mathrm{single}}
\newcommand{\pair}{\mathrm{pair}}

\renewcommand*{\le}{\leqslant}
\newcommand{\HALF}{\mathrm{HALF}}

\renewcommand*{\geq}{\geqslant}
\renewcommand*{\leq}{\leqslant}

\numberwithin{equation}{section}
\allowdisplaybreaks

\usepackage{hyperref}

\begin{document}
	\title{Hedonic Diversity Games}
	\author{Robert Bredereck}
	\affiliation{TU Berlin, Berlin, Germany}
	\author{Edith Elkind}
	\affiliation{University of Oxford, Oxford, U.K.}
	\author{Ayumi Igarashi}
	\affiliation{Kyushu University, Fukuoka, Japan}
\begin{abstract}
We consider a coalition formation setting where each agent belongs to one of the two types,
and agents' preferences over coalitions are determined
by the fraction of the agents of their own type in each coalition. This setting 
differs from the well-studied
Schelling's model in that some agents may prefer homogeneous coalitions, while others
may prefer to be members of a diverse group, or a group that mostly
consists of agents of the other type. We model this setting as a hedonic game
and investigate the existence of stable
outcomes using hedonic games solution concepts.
We show that a core stable outcome may fail to exist and checking the existence
of core stable outcomes is computationally hard. On the other hand, we propose an efficient algorithm
to find an individually stable outcome under the natural assumption that agents' preferences
over fractions of the agents of their own type are single-peaked. 
	\vspace{-10pt}
\end{abstract}

\maketitle

\section{Introduction}\label{sec:intro}
At a conference dinner, researchers split into groups to chat over food. Some junior researchers
prefer to stay in the company of other junior researchers, as they want to relax after a long
day of talks. Some senior researchers prefer to chat with their friends, who also happen 
to be senior researchers. But there are also junior researchers who want to use the dinner
as an opportunity to network with senior researchers, as well as senior researchers who are
eager to make the newcomers feel welcome in the community, and therefore want to talk
to as many junior people as possible. 
 
This example can be viewed as an instance of a coalition formation problem.
The agents belong to two types (senior and junior), and their preferences over coalitions
are determined by the fraction of agents of each type in the coalition.
This setting is reminiscent of the classic Schelling model~\cite{Schelling1969}, 
but there is an important difference: a standard assumption in the Schelling model is 
{\em homophily}, i.e., the agents are assumed to prefer to be surrounded by agents of their own
type, though they can tolerate the presence of agents of the other type, as long as their
fraction does not exceed a pre-specified threshold. In contrast, 
in our example some agents have homophilic preferences, 
while others have {\em heterophilic preferences}, 
i.e., they seek out coalitions with agents who are not like them. 

There is a very substantial body of research on homophily and heterophily in group formation.
It is well-known that in a variety of contexts, ranging from residential location~\cite{Schelling1969,Zhang2004a,Zhang2004b} to 
classroom activities and friendship relations~\cite{Moody2001,Newman2010}, 
individuals prefer to be together with those who are similar to them.
There are also settings where
agents prefer to be in a group with agents of the other type(s): for instance, in a coalition 
of buyers and sellers, a buyer prefers to be in a group with many other sellers
and no other buyers, so as to maximize their negotiating power. 
\citet{Aziz:2014:FHG,AzizBBHOP17} model this scenario as a {\em Bakers and Millers game}, 
where a baker wants to be in a coalition with many millers, whereas a miller wants to be in a coalition
with many bakers. However, there are also real-life scenarios
where agents can have different attitudes towards diversity:
this includes, for instance, language learning by immersion (with types being learners' native languages), 
shared accommodation (with types being genders), primary and secondary education
(with types being races and income groups), etc. In all these settings we expect the agents
to display a broad range of preferences over ratios of different types in their group.

\smallskip

\noindent{\bf Our contribution\ }
The goal of our paper is to provide strategic foundations 
for the study of coalition formation scenarios where each agent 
may have a different degree of homophily. More concretely, we consider settings where agents are divided 
into two types (blue and red), and each agent has preferences regarding the fraction of the agents of her own type,
which determine her preferences over coalitions. For most of our results, we assume that
agents' preferences are single-peaked, i.e., each agent has a preferred ratio $\theta$ of agents
of her own type, and prefers one fraction $\theta_1$ to another fraction 
$\theta_2$ if $\theta_1$ is closer to $\theta$ than $\theta_2$ is. 
Our model allows agents to express a variety of preferences including both complete homophily and 
complete heterophily. 

We model this setting as a hedonic game, and investigate the existence of stable outcomes
according to several established notions of stability, such as core stability, Nash stability 
and individual stability~\cite{Bogomolnaia2002,Banerjee2001}. 

We demonstrate that a core stable outcome may fail to exist, 
even when all agents have single-peaked preferences. Moreover, 
we show that deciding whether a core stable outcome exists is NP-complete.
However, we identify several interesting special cases where the core is guaranteed to be non-empty. 

We then consider stability notions that are defined in terms of individual deviations.
There is a simple reason why a Nash stable outcome may fail to exist. However, 
we propose an efficient algorithm to reach an individually stable outcome, i.e., 
an outcome where if some agent would like to deviate from her current coalition 
to another coalition, at least one agent in the target coalition would object to the move.
Our proof employs a careful and non-trivial adaptation of the algorithm of \citet{Bogomolnaia2002}
for single-peaked anonymous games.
Our algorithm is {\em decentralized} in the sense that, 
by following a certain set of rules, the agents can form a stable partition by themselves.

\smallskip

\noindent{\bf Related work\ }
Our work is related to the established line of research that studies 
the impact of homophily on residential segregation. The seminal paper 
of Thomas Schelling~\cite{Schelling1969} introduced a model of 
residential segregation in which two types of individuals are located 
on a line, and at each step a randomly chosen 
individual moves to a different location if the fraction of the 
like-minded agents in her neighborhood is below her tolerance ratio. 
With simple experiments using dimes and pennies, \citet{Schelling1969} 
found that any such dynamics almost always reaches a total segregation even 
if each agent only has a mild preference for her own type.

Following numerous papers empirically confirming 
Shelling's result \cite{Farley1993,Massey1994,Emeerson2001,Doff2007,ClarkFossett2008,Clark2008,Alba1993}, \citet{Young} was the first to provide a rigorous theoretical argument, 
showing that stability can only be achieved if agents are divided into homogeneous groups. 
In contrast, \citet{Brandt:stoc} showed that with tolerance parameter being 
exactly $\frac{1}{2}$, the average size of the monochromatic community is independent of the size of 
the whole system; subsequently, \citet{Immorlica:soda} extended this analysis to the two-dimensional 
grid. The recent work of \citet{ChauhanLM18} considers a more general utility function, taking into account 
both the fraction of the like-minded agents in the neighborhood and the location in which each agent
is situated; they investigate the existence of stable outcomes for several topologies, such as a grid 
and a ring.

We note, however, that our model is fundamentally different from  
Shelling's model, for at least three reasons. First, we do not assume any underlying topology 
that restricts coalition formation. Second, the coalitions in an outcome of a hedonic game 
are pairwise disjoint, while the neighborhoods in the Schelling model may overlap.
Finally, as argued above, our model does not assume homophilic preferences.

There is also a substantial literature on stability in hedonic games, starting with the early work
of \citet{Bogomolnaia2002}. Among the various classes of hedonic games, two classes
are particularly relevant for our analysis: fractional hedonic games~\cite{Aziz:2014:FHG,AzizBBHOP17}
and anonymous hedonic games \cite{Bogomolnaia2002}.

In fractional hedonic games the agents are located on a social network, 
and they prefer a coalition $C$ to a coalition $C'$
if the fraction of their friends in $C$ is higher than in $C'$.
The Bakers and Millers game is an example of a fractional hedonic game, where it is assumed
that each baker is a friend of each miller, but no two agents of the same type are friends. 
\citet{AzizBBHOP17,Aziz:2014:FHG}, and, subsequently, \citet{bilo1,bilo2} identify several special cases of 
fractional hedonic games where the set of core stable outcomes is non-empty.
In particular, \citet{Aziz:2014:FHG,AzizBBHOP17} give a characterization 
of the set of strictly core stable outcomes in the Bakers and Millers game.

In anonymous hedonic games agents' preferences over coalitions
depend on the size of these coalitions only. 
Similarly to our setting, it is known that with single-peaked anonymous preferences, 
there is a natural decentralized process to reach individual stability; 
however, a core stable outcome may fail to exist \cite{Bogomolnaia2002}, 
and deciding the existence of a core stable outcome is NP-complete \cite{Ballester2004}.

There are also other subclasses of hedonic games where stable outcomes are guaranteed to exist, such as 
acyclic hedonic games \cite{Demange2004,igarashi2016hedonicgraph}, dichotomous games 
\cite{peters2016dichotomous}, and top-responsive games \cite{Alcade2004}.


\section{Our Model}
For every positive integer $s$, we denote by $[s]$ the set $\{1, \dots, s\}$.
We start by defining the class of games that we are going to consider.

\begin{definition}
A {\em diversity game} is a triple $G = (R,B,(\succ_i)_{i \in R\cup B})$, 
where $R$ and $B$ are disjoint sets of {\em agents} 
and for each agent $i\in R\cup B$ it holds that
$\succ_i$ is a linear order over the set
$$
\Theta = \left\{\frac{r}{s}\mathrel{\Big|}  r\in\{0, 1, \dots, |R|\}, 
s\in\{1, \dots, |R|+|B|\} \right\}.
$$
We set $N=R\cup B$; the agents in $R$ are called the {\em red agents},
and the agents in $B$ are called the {\em blue agents}.
\end{definition}
We refer to subsets of $N$ as {\em coalitions}.
For each $i \in N$, we denote by $\calN_i$ the set of coalitions containing $i$. 
For each coalition $S \subseteq N$, 
we say that $S$ is {\em mixed} if it contains both red and blue agents; 
a mixed coalition $S$ is called a {\em mixed pair} if $|S|=2$. 

For each agent $i\in N$, we interpret the order $\succ_i$ as her preferences
over the fraction of the red agents in a coalition; for instance, if $\frac{2}{3}\succ_i \frac{3}{5}$,
this means that agent $i$ prefers a coalition in which two thirds of the agents are red
to a coalition in which three fifths of the agents are red. 

For each coalition $S$, we denote by $\theta_R(S)$ the fraction of the 
red agents in $S$, i.e., $\theta_R(S)=\frac{|S \cap R|}{|S|}$;
we refer to this fraction as the {\em red ratio of $S$}. 
For each $i \in N$ and $S, T \in \calN_i$, we say that 
agent $i$ {\em strictly prefers} $S$ to $T$ if 
$\theta_R(S) \succ_i \theta_R(T)$, and we say that 
$i$ {\em weakly prefers} $S$ to $T$ if $\theta_R(S)=\theta_R(T)$ 
or $\theta_R(S) \succ_i \theta_R(T)$. 
A coalition $S$ is said to be {\em individually rational} 
if every agent $i$ in $S$ weakly prefers $S$ to $\{i\}$.

An {\em outcome} of a diversity game 
is a partition of agents in $N$ into disjoint coalitions. 
Given a partition $\pi$ of $N$ and an agent $i \in N$, 
we write $\pi(i)$ to denote the unique coalition in $\pi$ that contains $i$. 
A partition $\pi$ of $N$ is said to be {\it individually rational} 
if all coalitions in $\pi$ are individually rational. 

The {\em core} is the set of partitions that are resistant to {\em group deviations}. 
Formally, we say that a coalition $S \subseteq N$ {\it blocks} 
a partition $\pi$ of $N$ if every agent $i \in S$ strictly prefers $S$ to her own coalition $\pi(i)$.
A partition $\pi$ of $N$ is said to be {\it core stable}, 
or in the {\em core}, if no coalition $S \subseteq N$ blocks $\pi$.

We also consider outcomes that are immune to {\em individual deviations}. 
Consider an agent $i \in N$ and a pair of coalitions 
$S\not\in\calN_i$ and $T\in\calN_i$. An agent $j \in S$ {\it accepts} 
a deviation of $i$ to $S$ if $j$ weakly prefers $S\cup \{i\}$ to $S$. 
A deviation of $i$ from $T$ to $S$ is said to be an {\em NS-deviation} 
if $i$ prefers $S \cup \{i\}$ to $T$; and an {\it IS-deviation} 
if it is an NS-deviation and all agents in $S$ accept it.
A partition $\pi$ is called {\em Nash stable} (NS) (respectively, {\em individually stable} (IS)) 
if no agent $i\in N$ has an NS-deviation
(respectively, an IS-deviation) from $\pi(i)$ to another coalition $S\in \pi$ or to $\emptyset$.

We say that the preferences $\succ_i$ of an agent $i \in N$ are {\em single-peaked} 
if for every $i \in N$ there is a peak $p_i \in [0,1]$ such that
\[
\theta_1 < \theta_2 \leq p_i~\mbox{or}~\theta_1 > \theta_2 \geq p_i \Rightarrow  \theta_2 \succ_i \theta_1. 
\]
In particular, if an agent has a strong preference for being in the majority, 
then her preferences are single-peaked, as illustrated in the following example.

\begin{example}[Birds of a feather flock together]\label{ex:homo}
Suppose that all agents in $R$ are smokers and all agents in $B$ are non-smokers.
Then we expect an agent to prefer groups with the maximum possible ratio
of agents of her own type. Formally, 
for each $r\in R$ and each $\theta, \theta'\in\Theta$
we have $\theta\succ_r \theta'$ if and only if $\theta>\theta'$, and
for each $b\in B$ and each $\theta, \theta'\in\Theta$
we have $\theta\succ_b \theta'$ if and only if $\theta<\theta'$. 
In this case, the partition in which each agent forms a singleton
coalition is core stable and Nash stable (and hence also individually stable).
\end{example}

If an agent strongly prefers to be 
surrounded by agents of the other type, her preferences are single-peaked, too.
\begin{example}[Bakers and Millers \citep{Aziz:2014:FHG}]\label{ex:hetero}
Suppose that each agent prefers the fraction of agents of the other type to be as high as possible. 
This holds, for instance, when individuals of the same type compete to trade with 
individuals of the other type. A {\em Bakers and Millers game} 
is a diversity game where for each $r\in R$ and each $\theta, \theta'\in\Theta$
we have $\theta\succ_r \theta'$ if and only if $\theta<\theta'$, and
for each $b\in B$ and each $\theta, \theta'\in\Theta$
we have $\theta\succ_b \theta'$ if and only if $\theta>\theta'$. 
Note that if $|R|=|B|$, a partition into mixed pairs is core stable and Nash stable;
indeed, \citet{Aziz:2014:FHG} prove that every Bakers and Millers game has a non-empty core.
\end{example}

\section{Core stability}
Examples~\ref{ex:homo} and~\ref{ex:hetero} illustrate that if all agents have 
extreme homophilic or extreme heterophilic preferences, the core is guaranteed to be non-empty.
However, we will now show that in the intermediate case the core may be empty, even if
all agents have single-peaked preferences.

\begin{example}\label{ex:emptycore}
Consider a diversity game $G$, where
the set of agents is given by 
$R=\{r_1, r_2,r_3, r_4,r_5,r_6,r_7\}$ and $B=\{b_1,b_2\}$. 
Agents can be divided into the following three categories 
with essentially the same preferences: 
$X=\{r_1, r_2,r_3, r_4,b_1\}$, 
$Y=\{r_5\}$, and 
$Z=\{r_6,r_7,b_2\}$. 
We have $\Theta=\{0, \frac{1}{3}$, $\frac{1}{2}$, 
$\frac{3}{5}$, 
$\frac{2}{3}$, $\frac{5}{7}$, 
$\frac{3}{4}$, $\frac{7}{9}$, $\frac{4}{5}$, $\frac{5}{6}$, $\frac{6}{7}$, $\frac{7}{8}, 1\}$. 
Each agent has the following single-peaked preferences over the ratios of red agents:
\begin{itemize}
\item $r_1, r_2,r_3, r_4: \frac{6}{7} \succ \frac{5}{6}  \succ \frac{4}{5} \succ \frac{7}{9} \succ \frac{3}{4}  \succ \frac{7}{8} \succ 1 \succ \frac{5}{7} \succ \cdots$
\item $b_1: \frac{6}{7} \succ \frac{5}{6}  \succ \frac{4}{5} \succ \frac{7}{9} \succ \frac{3}{4} \succ \frac{7}{8} \succ \frac{5}{7} \succ \cdots$
\item $r_5: \frac{5}{6} \succ \frac{4}{5} \succ \frac{7}{9} \succ \frac{3}{4} \succ \frac{6}{7} \succ \frac{7}{8}  \succ 1\succ \frac{5}{7} \succ \cdots$
\item $b_2: \frac{3}{4} \succ \frac{7}{9} \succ \frac{4}{5}  \succ \frac{5}{6} \succ \frac{6}{7} \succ \frac{7}{8} \succ \frac{5}{7} \succ \cdots $
\item $r_6,r_7: \frac{3}{4} \succ \frac{7}{9} \succ \frac{4}{5}  \succ \frac{5}{6} \succ \frac{6}{7} \succ \frac{7}{8} \succ 1 \succ \frac{5}{7} \succ \cdots $
\end{itemize}
Figure $1$ illustrates the preferences of each preference category of red agents. 
\end{example}

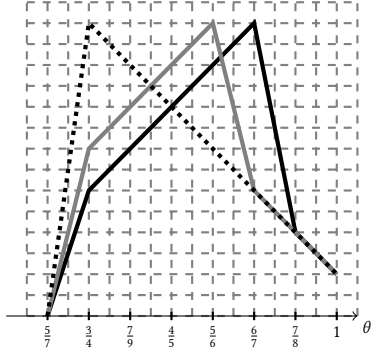
\begin{figure}[htb]\label{fig:singlepeak}
\begin{center}
\begin{tikzpicture}[scale=0.55, transform shape]
\draw[thick,color=gray,step=.5cm, dashed] (-0.5,0) grid (7.5,7.5);

\draw [thick] (0,-.1) node[below]{\Large $\frac{5}{7}$} -- (0,0.1);
\draw [thick] (1,-.1) node[below]{\Large $\frac{3}{4}$} -- (1,0.1);
\draw [thick] (2,-.1) node[below]{\Large $\frac{7}{9}$} -- (2,0.1);
\draw [thick] (3,-.1) node[below]{\Large $\frac{4}{5}$} -- (3,0.1);
\draw [thick] (4,-.1) node[below]{\Large $\frac{5}{6}$} -- (4,0.1);
\draw [thick] (5,-.1) node[below]{\Large $\frac{6}{7}$} -- (5,0.1);
\draw [thick] (6,-.1) node[below]{\Large $\frac{7}{8}$} -- (6,0.1);
\draw [thick] (7,-.1) node[below]{\Large $1$} -- (7,0.1);

\draw[->] (-1,0) -- (7.5,0) node[below right] {\Large $\theta$};
\draw[-,color=black,ultra thick] (0,0) -- (1,3) -- (2,4) -- (3,5) -- (4,6) -- (5,7) -- (6,2) -- (7,1);
\draw[-,color=gray,ultra thick] (0,0) -- (1,4) -- (2,5) -- (3,6) -- (4,7) -- (5,3) -- (6,2) -- (7,1);
\draw[-,color=black,ultra thick,dotted] (0,0) -- (1,7) -- (2,6) -- (3,5) -- (4,4) -- (5,3) -- (6,2) -- (7,1);
\end{tikzpicture}
\caption{Single-peaked preferences over the ratios of red agents. 
The thick, grey, and dotted lines represent the preferences of agents in $X$, $Y$, and $Z$, 
respectively. We omit the ratios $\frac{1}{3}, \frac{1}{2}, \frac{2}{3}$ 
since there is no individually rational coalition with these fractions of red agents.}
\end{center}
\end{figure}

\noindent We will now argue that the game in Example~\ref{ex:emptycore} has empty core.
\begin{proposition}
The game $G$ in Example~\ref{ex:emptycore} has no core stable outcomes. 
\end{proposition}
\begin{proof}
Suppose towards a contradiction that there exists a core stable outcome $\pi$. We note that, 
by individual rationality, $\theta_R(S) \geq 3/4$ for every mixed coalition $S \in \pi$, 
as red agents in a coalition $S$ with $\theta_R(S)\le 5/7<3/4$ would strictly prefer to be alone. 
Also, $\pi$ contains at least one mixed coalition, 
as otherwise $\{r_1, r_2, r_3, b_1\}$ would block $\pi$.

Let $\theta^*$ be the largest ratio of red agents in a mixed coalition in $\pi$. 
We have argued that $\theta^*\geq 3/4$. 
Also, if $3/4 \leq \theta^* \leq 4/5$, 
then $\theta_R(\pi(x)) \leq 4/5$ for all $x \in X$ 
and $3/4 \leq \theta_R(\pi(r_5)) \leq 4/5$ or 
$\theta_R(\pi(r_5))=1$; thus coalition $X \cup Y$ with $\theta_R(X\cup Y) = 5/6$ blocks $\pi$. 
Hence, $\theta^* \geq 5/6$. Now since $\pi$ contains at least one mixed coalition 
with red ratio at least $5/6$, and all mixed coalitions in $\pi$ must have red ratio at least $3/4$, 
if there is more than one mixed coalition, the number of agents would be at least 10, 
a contradiction. Thus, $\pi$ contains exactly one mixed coalition of red ratio at least $5/6$. 
It follows that for each agent $i$ it holds that either
\begin{itemize}
\item $i$ belongs to a mixed coalition of red ratio at least $5/6$, 
  i.e., $\theta_R(\pi(i)) \geq 5/6$; or
\item $i$ belongs to a completely homogeneous coalition, i.e., if $i \in Y$, 
  then $\theta_R(\pi(i))=0$ and if $i \in R$, then $\theta_R(\pi(i))=1$. 
\end{itemize}
In particular, this means that $\theta_R(\pi(z)) \neq 3/4$ for all $z \in Z$, 
and hence each $z \in Z$ prefers $3/4$ to $\theta_R(\pi(z))$.

Now suppose that $r_5$ does not belong to a coalition with his favorite red ratio, i.e., $5/6$. 
Then $\theta_R(\pi(r_5)) \geq 6/7$ and thus $r_5$ prefers $3/4$ to $\theta_R(\pi(r_5))$. 
Thus, the coalition $Y \cup Z$ of red ratio $3/4$ blocks $\pi$, a contradiction. 
Hence, $r_5$ belongs to a coalition of his favorite red ratio $5/6$, and thus $\theta^*=5/6$. 
Further, if some agent $r \in X\cap R$ does not belong to a mixed coalition, 
then $\theta_R(\pi(r))=1$ and the coalition $\{r\} \cup Z$ would block $\pi$. 
Hence, the unique mixed coalition of red ratio $5/6$ contains both $r_5$ 
and all four red agents in $X$, which means that no red agent in $Z$ 
belongs to the mixed coalition. Now we have: 
\begin{itemize}
\item $b_1 \in X \cap B$ prefers $6/7$  to $\theta_R(\pi(b_1))$ since 
   $\theta_R(\pi(b_1))=0$ or $\theta_R(\pi(b_1))=5/6$; and
\item each $x \in X \cap R$ prefers $6/7$ to $\theta_R(\pi(x))=5/6$; and
\item each $z \in Z \cap R$ prefers $6/7$ to $\theta_R(\pi(z))=1$. 
\end{itemize}
Then coalition $X\cup (Z\cap R)$ of red ratio $6/7$ would block $\pi$, a contradiction. 
We conclude that $G$ does not admit a core stable partition.
\end{proof}

Indeed, we can show that checking whether a given diversity game has a non-empty core
is NP-complete.

\begin{theorem}\label{thm:core-np}
The problem of checking whether a diversity game $G=(R, B, (\succ_i)_{i\in R\cup B})$ 
has a non-empty core is {\em NP}-complete. 
\end{theorem}
\begin{proof}
Checking core stability can be done in polynomial time and hence our problem is in NP. Indeed, for each size of a coalition $s \in [n]$ and each number $r \in \{0\} \cup [|R|]$, we can check in polynomial time if there is a blocking coalition of size $s$ including $r$ red players: We consider the set $S_{s,r}$ of all players who strictly
prefer the ratio $\frac{r}{s}$ to their ratio under $\pi$ and  verify whether $S_{s,r}$ admits a subset of size at least $s$ that contains $r$ red players; if this is the case, there is a coalition of size exactly $s$ that strongly blocks $\pi$. If no such coalition exists, $\pi$ is core stable. 

We reduce from the problem of deciding whether the core of an anonymous hedonic game is non-empty.
 A hedonic game is called {\em anonymous} if all players only care about the sizes of their coalition. Formally, an {\em anonymous hedonic game} is a pair $(N,(>_i)_{i \in N})$, 
where $N$ is a set of {\em agents}  and for each agent $i\in N$ it holds that
$>_i$ is a linear order over the set~$[|N|]$ of coalition sizes.  
 Ballester~\cite{Ballester2004} has shown that deciding whether the core of an anonymous hedonic game is non-empty is NP-hard.
 
 Let $n$~be the number of players and let the preferences of the players be encoded as follows
 $$a_i: s(i,1) \succ s(i,2)  \succ \dots \succ s(i,n),$$
 where $s(i,\ell)$ denotes $\ell$th most preferred the coalition size of agent~$i$.\footnote{For
 technical reasons, we assume that always all sizes between $1$ and $n$ are listed, including irrational ones.}
 We construct our equivalent hedonic diversity game instance as follows.
 First, for each original player we introduce one blue agent
 $$b_i: \frac{1}{s(i,1) + 1} \succ \frac{1}{s(i,2) + 1} \succ \dots \succ \frac{1}{s(i,N) + 1} \succ 0.$$
 Naturally, the preferences of the original player~$a_i$ with respect to coalitions of size~$s$
 shall be represented by the preference of player~$b_i$ with respect to coalitions with~$s$ blue players and exactly one red player.
 For each possible coalition size $1 \le t \le n$ from the original game, we introduce $\lceil n/t \rceil + 1$ red players~$r_t^\ell$ that would accept
 to be in a coalition with~$t$ blue player and otherwise prefer to be alone:
 $$r_t^\ell: \frac{1}{t + 1} \succ 1, \forall 1 \le \ell \le \lceil n/t \rceil + 1$$
 
 If we could guarantee that every mixed coalition has at most one red player, then, the correctness proof would be almost immediate.
 It might, however, happen that there are coalitions with more then one red player.
 In the following correctness proof, we will see that these coalitions are unproblematic.
 
 Formally, we show that the core of the original anonymous hedonic game is non-empty if and only if the core of the newly constructed
hedonic diversity game is non-empty.
 
 For the first direction, assume that the core of the anonymous game was non-empty.
 That means there is a partition~$\pi=\{P_1$, $P_2$, $\dots$, $P_k\}$ of the original players into coalitions so that no blocking coalition exists.
 We build a core-stable partition~$\pi'=\{P'_1$, $P'_2$, $\dots$, $P'_k$, $P_R\}$ of our newly constructed hedonic diversity game instance from~$\pi$ as follows:
 $P'_j=\{b_i \mid a_i \in P_j\} \cup \{r_t^\ell\}$, where $t=|P_j|$ and $\ell=|\{P_x \mid x \le j \wedge |P_x|=t\}|$
 and $P_R=R\setminus(\bigcup P_j)$ simply contains all remaining red players. 
 
 Assume towards a contradiction that there was a blocking coalition~$C^*$ of players that blocks~$\pi'$.
 We first argue that we can assume w.l.g.\ that~$C^*$ contains at most one red player.
 Of course, $C^*$ cannot contain only red players because every red player in $\pi'$ is either in $P_R$ and, thus, already only with other red players
 or it is by definition of~$\pi'$ in one of its most preferred coalition.
 Now, assume $C^*$~contains $y$~red players.
 By the definition of the preferences of the red players, $|C^*|$ must be the red player's most preferred coalition ratio.
 That is, the number of blue players in~$C^*$ must be divisible by~$y$.
 This immediately implies that there is a smaller blocking coalition~$C^{**}$ which consists of one aribitrary red player from~$C^*$
 and $(|C^*|-y)/y$ arbitrary blue players from~$C^*$.
 Finally, $C^*$ (now containing exactly one red player) can only be blocking because every blue player prefers 
 a coalition with one red and $|C^*|-1$ blue players towards its current coalition in~$\pi'$.
 The way we constructed the preferences of of blue players ensures that this implies that the coalition $C^A=\{a_i \mid b_i \in C^*\}$
 must be a blocking coalition for $\pi$---a contradiction.
 
 For the second direction, assume that the core of the diversity game was non-empty.
 That means there is a partition~
 \[
 \pi'=\{P'_1, P'_2, \dots, P'_k, P'_R\} 
 \]
 of the blue and red players into coalitions so that no blocking coalition exists.
 
 We assume w.l.g.\ that all coalitions~$P'_j$ are mixed and~$P'_R$ contains only red players.
 First, there could exist multiple purely red coalitions but merging them into one gives a partition
 that must be core stable if and only if the original partition was.
 Second, we can easily observe that $P'_R$ is non-empty, because
 no player prefers mixed coalitions with a majority of red players
 and there are much more red than blue players.
 Third, a purely blue coalition cannot exist.
 Every blue player prefers coalitions with as many red as blue players to those coalitions 
 and there are always enough red players available that would be willing to pair up.
 (Since one of the $s(i,\ell)$ equals $1$, ratio $\frac{1}{2}$ is preferred to $0$.)
 Finally, for each integer $1 \le t \le n$ coalition $P_R$ must contain at least one
 red player~$r_t^\ell$ which would prefer to be in a coalition with $t$~blue players.
 Let us refer to this number as~$\ell(t)$.
 
 Furthermore, we can assume that w.l.g.\ each $P'_j$ contains at exactly one red player.
 Every red player that is contained in some~$P'_j$ must be in a coalition
 that has the player's most preferred ratio of red player, otherwise the player would be in~$P_R$.
 In particular, this means that if there are $y$~red players in~$P'_j$, then all these players
 must be of the same type and prefer coalitions with a ratio of $\frac{y}{y(t+1)}$ red players
 for some integer~$t$.
 Clearly, splitting up $P'_j$~ into~$y$ coalitions with exactly one red player and exactly $t$~blue players
 will give us another core stable partition, because these new coaltions have exactly the
 same ratio as~$P'_j$ has.
 
 With all these assumptions, we are now ready to define a partition~$\pi=\{P_1,\dots,P_k\}$ for the
 players of the original anonymous hedonic game instance via $P_j=\{a_i \mid b_i \in P'_j\}$.
 Finally, assume towards a contradiction that there would be a coalition~$C^A$ that blocks $\pi$.
 Then, using our assumptions on $\pi'$ and the definition of the players preferences,
 the coalition~$C^*=\{b_i \mid a_i \in C^A\} \cup r_t^\ell$, where $t=|C^A|$ and $\ell=\ell(t)$,
 would be a blocking coalition for $\pi$:
 Player~$r_t^\ell$ comes from $P_R$ and, hence, clearly prefers~$C^*$ and the blue players,
 by definition of their preferences, prefer $C^*$~if and only if the corresponding original
 players prefer a size~$t$ coalition towards their current ones---a contradiction that there
 is a blocking coalition for~$\pi$.
\end{proof}

In Example~\ref{ex:emptycore} there are at least two agents of each type. 
In contrast, if one of the types is represented by a single agent,
then the core is guaranteed to be non-empty.

\begin{proposition}\label{prop:core}
Let $G = (R,B,(\succ_i)_{i \in R \cup B})$ be a diversity game with $|R|=1$ or $|B|=1$. 
Then the core of $G$ is non-empty, and a partition in the core can be constructed in polynomial time.
\end{proposition}
\begin{proof}
Without loss of generality, assume that $R=\{r\}$. Let $\Theta_r$ be the 
set of ratios of red agents in individually rational coalitions to which $r$ belongs, i.e., 
$$
\Theta_r=\{ \theta_R(S) \mid r \in S \subseteq R \cup B~\land~\mbox{$S$ is individually rational} \}.
$$ 
Take $r$'s favorite ratio $\theta^*$ among the fractions in $\Theta_r$. Construct a coalition $S^*$ 
with $\theta_R(S^*)=\theta^*$ and put the rest of blue agents into singletons. Clearly, the resulting 
partition $\pi$ is not blocked by any coalition of blue agents. Also, if there exists a coalition $S$ 
that contains the red agent $r$ and blocks $\pi$, then $\theta_R(S) \succ_r \theta^*$ and $S$ is 
individually rational, contradicting the choice of $\theta^*$. Thus, the resulting partition is core 
stable.
\end{proof}

We also note that in the game in Example~\ref{ex:emptycore} agents' preferences
belong to one of the three categories. The next proposition shows that if all agents have the 
same preferences, then there is a core stable outcome. We conjecture that with only two types of 
single-peaked preferences, the core is non-empty as well.
\begin{proposition}
Let $G = (R,B,(\succ_i)_{i \in R \cup B})$ be a diversity game such that each 
agent has the same preference over $\Theta\setminus\{0, 1\}$. Then $G$ has a non-empty core
and a partition in the core can be constructed in polynomial time.
\end{proposition}
\begin{proof}
We claim that our game satisfies the {\em top coalition property}: for each $S \subseteq N$, there 
exists a {\em top coalition} $T \subseteq S$ such that every agent in $T$ 
weakly prefers $T$ to every other coalition contained in $S$. 
Indeed, take any $S \subseteq N$. If there is an agent 
$i \in S$ who weakly prefers his own coalition $\{i\}$ to every mixed coalition contained in $S$, then 
taking $T=\{i\}$ certifies the existence of a top coalition. Thus, suppose otherwise, i.e., every agent 
in $S$ strictly prefers some mixed coalition to his own singleton. As every agent has the same 
ranking over the fractions of red agents in mixed coalitions, there is a mixed coalition 
$T \subseteq S$ such that $\theta_R(T)$ is the favorite red ratio for every agent in $T$ 
among the ratios of red agents in the subsets of $S$, which means that $T$ is a top coalition.
This argument also gives us an efficient algorithm for identifying a top coalition:
it suffices to check whether some singleton forms a top coalition, and, if not, 
find the agents' most preferred ratio that can be implemented in $\Theta\setminus\{0, 1\}$
and construct a coalition with this red ratio.

We can now construct a stable partition by repeatedly identifying a top coalition
with respect to the current set of agents, adding it to the partition, and removing the agents
in that coalition from the current set of agents; it is clear that this procedure  
can be executed in polynomial time, and \citet{Banerjee2001} argue that it produces a core stable outcome.
\end{proof}


\section{Nash Stability and Individual Stability}
We have seen that core stability may be impossible to achieve. It is 
therefore natural to ask whether every diversity game has an outcome 
that is immune to individual deviations. It is easy to see that the answer is `no'
if we consider NS-deviations, even if we restrict ourselves 
to single-peaked preferences: for instance, the game where 
there is one red agent who prefers to be 
alone and one blue agent who prefers to be in a mixed coalition
has no Nash stable outcomes. 

In contrast, each diversity game with single-peaked preferences admits
an individually stable outcome. Moreover, such an outcome can be computed
in polynomial time. In the remainder of this section, we will present
an algorithm that achieves this, and prove that it is correct.
The algorithm will be divided into three parts:

\begin{enumerate}
\item For agents with peaks greater than half, make mixed coalitions with red majority. 
      For agents with peaks smaller than half, make mixed coalitions with blue majority.
\item Make pairs from the remaining red agents and blue agents who are not in the mixed coalitions.
\item Put all the remaining agents into singletons.
\end{enumerate}
We will first show that one can construct a sequence of mixed coalitions with red majority 
that are immune to IS-deviations. 

\subsection{Create mixed coalitions with red majority}\label{sec:mixed}
In making mixed coalitions, we will employ a technique that is similar 
to the algorithm for anonymous games proposed by \citet{Bogomolnaia2002}. 
Intuitively, imagine that red agents and blue agents with peaks at least $1/2$
form two lines, each of which is ordered from the highest peak to the lowest peak. 
The agents enter a room in this order, 
with a single blue agent entering first and red agents successively joining it 
as long as the fraction of red agents does not exceed the minimum peak of the agents already in the room. 
Once the fraction of the red agents reaches the minimum peak, 
a red agent who enters the room may deviate to a coalition that has been formed before. 
We alternate these two procedures as long as there is a red agent 
who can be added without exceeding the minimum peak. If no red agent can enter a room, 
then agents start entering another room and create a new mixed coalition in the same way. 
The algorithm terminates if either all red agents or all blue agents with their peaks 
at least half join a mixed coalition. 
Figure~\ref{fig:mixed-coalitions} illustrates this coalition formation process. 
We formalize this idea in Algorithm~\ref{alg:sub:IS}. 
We will create mixed coalitions containing exactly one blue agent, 
so we define the {\em virtual peak} $q_i$ to be the favorite ratio of agent $i$ 
among the ratios of red agents in coalitions containing exactly one blue agent.

\begin{figure*}[tb]
	\centering
	\begin{tikzpicture}[scale=0.8, transform shape, every node/.style={minimum size=5mm, inner sep=1pt}]	
	
	\draw (-0.6,0.4) rectangle (3,-0.4);
	\draw[-, >=latex,dotted,thick] (0.4,0) -- (1.2,0);
	\node[draw, circle,fill=gray!50] at (0,0) {};
	\node[draw, circle,fill=gray!50] at (1.6,0) {};
	\node[draw, circle,fill=gray!50] at (2.4,0) {};
	\node at (1.2,-0.8) {\Large $S_0$};
	
	\begin{scope}[shift={(-4,0)}]
	\draw (-0.6,0.4) rectangle (3,-0.4);
	\node[draw, circle,fill=gray!50] at (0,0) {};
	\draw[-, >=latex,dotted,thick] (0.4,0) -- (1.2,0);
	\node[draw, circle,fill=gray!50] at (1.6,0) {};
	\node[draw, circle] at (2.4,0) {$b_1$};
	\node at (1.2,-0.8) {\Large $S_1$};
	\end{scope}
	
	\begin{scope}[shift={(-10,0)}]
	\draw[-, >=latex,dotted,thick] (3.4,0) -- (5,0);
	\draw[-, >=latex,dotted,thick] (0.4,0) -- (1.2,0);
	\draw (-0.6,0.4) rectangle (3,-0.4);
	\node[draw, circle,fill=gray!50] at (0,0) {};
	\node[draw, circle,fill=gray!50](2) at (1.6,0) {};
	\node[draw, circle] at (2.4,0) {$b_{t}$};
	\node at (1.2,-0.8) {\Large $S_{t}$};
	\end{scope}
	
	\begin{scope}[shift={(-14.5,-0.5)}]
	\node[draw, circle] at (0,0) {};
	\node[draw, circle] at (1,0) {};
	\node[draw, circle] at (2,0) {};
	\node[draw, circle] at (3,0) {};
	\draw[-, >=latex,dotted,thick] (-0.5,0) -- (-1.2,0);
	\end{scope}
	
	\begin{scope}[shift={(-14.5,0.5)}]
	\node[draw, circle,fill=gray!50] at (0,0) {};
	\node[draw, circle,fill=gray!50] at (1,0) {};
	\node[draw, circle,fill=gray!50] at (2,0) {};
	\node[draw, circle,fill=gray!50](1) at (3,0) {$r_i$};
	\draw[-, >=latex,dotted,thick] (-0.5,0) -- (-1.2,0);
	\end{scope}
	
	\draw [->, thick] (1.north) to [out=70,in=90] (-10,0.3);
	
	\draw [->, thick] (2.north) to [out=70,in=110] (-6,0.3);
	\end{tikzpicture}
	\vspace{-8pt}
	\caption{Illustration of how HALF$(R,B,(\succ_i)_{i \in R\cup B})$ creates 
                coalitions with red majority. Gray circles correspond to red agents, 
                while white circles correspond to blue agents.
		\label{fig:mixed-coalitions}
	}
\end{figure*}
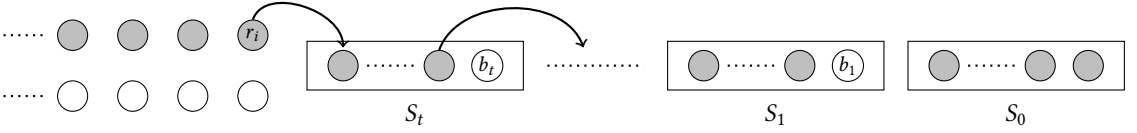

\begin{algorithm}
\SetKwInOut{Input}{input} 
\SetKwInOut{Output}{output}
\SetKw{And}{and}
\SetKw{None}{None}
\caption{HALF$(R,B,(\succ_i)_{i \in R\cup B})$}\label{alg:sub:IS}
\Input{A single-peaked diversity game $(R,B,(\succ_i)_{i \in R\cup B})$} 
\Output{$\pi$}
	sort red and blue agents so that $q_{r_1} \geq q_{r_2} \geq \ldots q_{r_{x}}$ and 
                                         $q_{b_1} \geq q_{b_2} \geq \ldots  \geq q_{b_y}$\;
	initialize $i\leftarrow 1$, $k \leftarrow 1$ and $S_0 \leftarrow \emptyset$\;
	initialize $R' \leftarrow \{\, r \in R \mid p_r \geq \frac{1}{2}\,\}$ and 
                   $B' \leftarrow \{\, b \in B \mid p_b \geq \frac{1}{2}\,\}$\;
	\While{$R' \neq \emptyset$ and $B' \neq \emptyset$}{
	set $S_k \leftarrow \{b_k\}$\;
	\While{$\theta_R(S_k \cup \{r_i\}) \leq  \min \{q_{r_{i}},q_{b_k}\}$, or 
           there exist a agent $r \in S_k \cap R$ and $t<k$ such that 
               $r$ has an IS-deviation from $S_k$ to $S_{t}$ \label{step:if}}{
	{\tt //  add red agents to $S_k$ as long as the ratio of red agents 
                  does not exceed the minimum virtual peak}\;
	\While{$\theta_R(S_k \cup \{r_i\}) \leq \min \{q_{r_{i}},q_{b_k}\}$\label{step:red:while}}{
	set $S_k \leftarrow S_k \cup \{r_i\}$\label{step:join}\;
	set $R' \leftarrow R' \setminus \{r_i\}$ and $i \leftarrow i+1$\;
	}	
	{\tt //  let red agents in $S_k$ deviate to smaller-indexed coalitions}\;
	\If{there exists an agent $r \in S_k \cap R$ and $t<k$ such that 
              $r$ has an IS-deviation from $S_k$ to $S_{t}$ \label{step:leave:if}}{
	choose $r$ and $S_t$ so that $\theta_R(S_{t} \cup \{r\})$ is 
            $r$'s most preferred ratio among the coalitions $S_{t}$ satisfying the above\;
	set $S_{t} \leftarrow S_{t} \cup \{r\}$, $S_k \leftarrow S_{t} \setminus \{r\}$\label{step:leave}\;
	}
	}
	set $B' \leftarrow B' \setminus \{b_k\}$, and $k \leftarrow k+1$\;
	}
	{\tt //  let the remaining agents deviate to mixed coalitions as long as 
                  they prefer the ratio of the deviating coalition to half}\;
	\While{there is an agent $r \in R'$ and a coalition $S_t$ such that $t\geq 0$, 
                  $\theta_R(S_{t} \cup \{r\}) \succ_r \frac{1}{2}$, 
         and all agents in $S_{t}$ accept a deviation of $r$ to $S_{t}$ \label{step:outside:if}}{
	choose $r$ and $S_{t}$ so that $\theta_R(S_{t} \cup \{r\})$ is 
            $r$'s most preferred ratio among the coalitions $S_{t}$ satisfying the above\;
	$S_{t} \leftarrow S_{t} \cup \{r\}$ and $R' \leftarrow R'\setminus \{r\}$ \label{step:outside}\;
	}
	\If{$R'=\emptyset$ and $S_k$ consists of a single blue agent \label{step:single:woman}}{
	return $\pi=\{\,\{r\} \mid r \in S_0 \,\} \cup \{S_1,S_2,\ldots,S_{k-1}\}$\label{step:blue}\;
	}
	\Else{
	return $\pi=\{\,\{r\} \mid r \in S_0 \,\} \cup \{S_1,S_2,\ldots,S_k\}$\;
	}
\end{algorithm}

In what follows, we assume that $S_0,S_1,\ldots,S_k$ are the final coalitions that have been obtained
at the termination of the algorithm, and that $\pi$ is the output returned by the algorithm. For each 
$t>0$ and $i \in S_{t}$,
\begin{itemize}
\item $i$ is called a {\em default agent} of $S_{t}$ if $i$ is a blue agent or $i$ is a red agent 
      who has joined $S_{t}$ at Step \ref{step:join}; 
\item $i$ is called a {\em new agent} of $S_{t}$ if $i$ is a red agent 
      who has joined $S_{t}$ in Step \ref{step:leave} or Step \ref{step:outside}.
\end{itemize}
We denote by $D_{t}$ the set of default agents of $S_{t}$, 
and we denote by $N_{t}$ the set of new agents of $S_{t}$. 

For each $S_t$ with $t>0$, each agent in $S_t$ is either a default agent of a new agent. Notice that $S_0$ starts with the empty set and plays the role of the {\em last resort} option for red agents, i.e., red agents can always deviate to $S_0$ if they strictly prefer staying alone to the mixed coalition they have joined. 

\noindent {\bf Deviations of blue agents\ } 
We will first show that the algorithm constructs a sequence of coalitions such that no blue agent in the 
coalitions has an IS-deviation to other coalitions. We establish this by proving a sequence of claims. 
First, it is immediate that the red ratio of each coalition $S_{t}$ is at least half, except for the 
last coalition, which may contain a single blue agent. 

\begin{lemma}\label{lem:half}
For each $S \in \pi$, the red ratio of $S$ is at least half, i.e., $\theta_R(S) \geq \frac{1}{2}$.
\end{lemma}
\begin{proof}
Take any $S \in \pi$. The claim is immediate when $S\subseteq S_0$. 
Suppose that $S=S_{t}$ for some $t>0$. Then it is clear that $S_{t}$ contains exactly one blue agent. 
If $S_{t}$ contains no red agent, then this would mean that all red agents with their peaks 
at least half joined $S_{t}$, but deviated to smaller indexed-coalitions, 
meaning that $S_{t}$ is the last coalition that has been formed, i.e., $t=k$; 
but this contradicts the construction in Step \ref{step:blue}. 
Thus, $S_{t}$ contains exactly one blue agent and at least one red agent, 
which implies that $\theta_R(S_{t}) \geq \frac{1}{2}$.
\end{proof}

This leads to the following lemma.

\begin{lemma}\label{lem:singleton}
For each $r \in S_0$, $r$ is a red agent and $r$ strictly prefers $\{r\}$ to a mixed pair.
\end{lemma}
\begin{proof}
By construction every agent in $S_0$ is a red agent. 
Suppose that before joining $S_0$, $r$ has joined a mixed coalition 
$S_{t}$ with $t > 0$ at Step \ref{step:join} and then deviated to $S_0$ 
at Step \ref{step:leave}. Just before $r$ has left $S_{t}$, 
the ratio $\theta'$ of red agents in $S_{t}$ is at least half, 
implying that $\theta' \succ_r \frac{1}{2}$ or $\theta' = \frac{1}{2}$. 
Since $r$ joined $S_0$ at Step \ref{step:leave}, we also have $1 \succ_r \theta'$. 
Combining these yields $1 \succ_r \frac{1}{2}$. The claim is immediate when $r$ joined $S_0$ 
at Step \ref{step:outside}.
\end{proof}

We also observe that, by the construction of the algorithm, all default agents belong to the coalition whose 
red ratio is at most their virtual peak; see Figure \eqref{fig:default} for an illustration.

\begin{lemma}\label{lem:default}
For every default agent $i$ in $S_{t}$, where $S_{t} \in \pi$ with $t>0$, 
\begin{enumerate}
\item the red ratio of $S_t$ is at most $i$'s virtual peak, i.e., $\theta_R(S_{t}) \leq q_i$; and 
\item $i$ weakly prefers $S_{t}$ to $\{i\}$ and to a mixed pair.
\end{enumerate}
\end{lemma}
\begin{proof}
Take any $t>0$ and take any default agent $i$ in $S_{t}$. 
Before the algorithm starts forming the next coalition $S_{t+1}$, 
the ratio $\theta_R(S_{t})$ remains below $q_i$ by the {\bf while}-condition 
in Step \ref{step:if}. After $S_{t+1}$ starts being formed, $\theta_R(S_{t})$ 
can only increase by accepting red agents from larger-indexed coalitions. 
However, if a deviation of some red agent to $S_{t}$ increases 
the fraction above $q_i$, this would mean that $i$ would not be willing 
to accept such a deviation. Thus, $\theta_R(S_{t}) \leq q_i$. 
To show the second statement, recall that $p_i \geq \frac{1}{2}$ 
by construction of the algorithm and $\theta_R(S_{t}) \geq \frac{1}{2}$ by Lemma~\ref{lem:half}. 
If $i$ is a blue agent, then by single-peakedness $i$ weakly prefers $S_t$ both to her own singleton 
and to a mixed pair. If $i$ is a red agent who has joined $S_{t}$ at Step \ref{step:join}, 
then by single-peakedness and the fact that he has not deviated to $S_0$ at Step \ref{step:leave},
 $i$ weakly prefers $S_t$ both to his singleton and to a mixed pair. 
\end{proof}

We are now ready to prove that no coalition in $\pi$ admits a deviation of a blue agent.

\begin{lemma}\label{lem:blue}
For each $S \in \pi$, there is an agent in $S$ who does not accept a deviation of a blue agent to $S$. 
\end{lemma}
\begin{proof}
If $S \subseteq S_{0}$ and $|S|=1$, it follows from Lemma \ref{lem:singleton} that no $r \in S$ accepts a deviation of a blue agent.
Suppose $S=S_{t}$ for some $t>0$, and assume towards a contradiction that some blue agent $b \in B$ can be accepted by all agents in $S_{t}\in \pi$ for some $t>0$. 
By Lemma \ref{lem:half}, the ratio of red agents in $S_{t}$ is at least $\frac{1}{2}$. If the fraction of red agents in $S_{t}$ is at most the peak of some agent, i.e., $\theta_R(S_{t}) \leq p_i$ for some agent $i \in S_{t}$, this means that
\[
\theta_R(S_{t}\cup \{b\})< \theta_R(S_{t}) \leq p_i, 
\]
i.e., $i$ would not accept $b$, a contradiction. 
Hence, suppose that the fraction of red agents in $S_{t}$ is beyond the maximum peak, 
i.e., $\theta_R(S_{t})> \max_{i \in S_{t}}p_i$. Take $i \in S_{t} \cap B$. 
By Lemma~\ref{lem:default}, 
\[
p_i < \theta_R(S_{t}) \leq q_i.
\]
Since $p_i \geq \frac{1}{2}$ by construction of the algorithm, we have $\theta_R(S_{t})>\frac{1}{2}$, 
which means that coalition $S_{t}$ contains one blue agent and at least two red agents. 
Hence, even if one red agent leaves the coalition, its red ratio would be at least half, i.e.,
\[
\frac{1}{2} \leq \frac{|S_{t} \cap R|-1}{|S_{t}|-1} (:=\theta^*) < p_i, 
\]
where the second inequality holds, 
since otherwise $p_i \leq \theta^*<\theta_R(S_{t}) \leq q_i$, 
which contradicts the fact that $q_i$ is $i$'s favorite ratio 
of the coalitions containing exactly one blue agent. 
Further, agent $i$ prefers $\theta_R(S_{t})$ to $\theta^*$ 
by single-peakedness and by the fact that $\theta^* < \theta_R(S_{t}) \leq q_i$. 
Since $i$ accepts the deviation of $b$ to $S_{t}$, after adding $b$ to $S_{t}$,  
the red ratio should remain at least $\theta^*$, 
implying that $\theta^* < \theta_R(S_{t}\cup \{b\})$. But this would mean that 
$\frac{|S_{t} \cap R|-1}{|S_{t}|-1} < \frac{|S_{t} \cap R|}{|S_{t}|+1}$, 
or, equivalently, $(|S_t|+1)(|S_{t} \cap R|-1)<  |S_{t} \cap R| (|S_{t}|-1)$.
This inequality can be simplified to $2 |S_t \cap R| < |S_t| +1$, 
implying $\theta^*< \frac{1}{2}$, a contradiction.
\end{proof}

\begin{figure}[htb]\label{fig:deviation}
\begin{subfigure}[t]{0.5\columnwidth}
\centering
\begin{tikzpicture}[scale=0.3, transform shape]
\draw[thick,gray] (-2,-1.5) rectangle (10,7.5);
\draw[->] (-1,0) -- (8.5,0) node[below right] {\Huge $\theta$};
\draw[-,color=gray,thick] (0,1) -- (1,2) -- (2,3) -- (3,4) -- (4,5) -- (5,6) -- (6,7) -- (8,1);
\draw [thick] (6,-.1) node[below]{\Huge $q_i$} -- (6,0.1);
\draw[-,dotted,color=gray] (6,0) -- (6,7);
\node [circle,fill=black] at (3,4) {};
\draw [thick] (3,-.1) node[below]{} -- (3,0.1);
\draw[-,dotted,color=gray] (3,0) -- (3,4);
\draw[->,thick] (3,4) -- (4,5);
\draw [thick] (4,-.1) node[below]{\Huge $\theta_R(S_{t})$} -- (4,0.1);
\draw[-,dotted,color=gray] (4,0) -- (4,5);
\end{tikzpicture}
\caption{Virtual peak of a default agent.}\label{fig:default}
\end{subfigure}%
\begin{subfigure}[t]{0.5\columnwidth}
\centering
\begin{tikzpicture}[scale=0.3, transform shape]
\draw[thick,gray] (-2,-1.5) rectangle (10,7.5);
\draw[->] (-1,0) -- (8.5,0) node[below right] {\Huge $\theta$};
\draw[-,color=gray,thick] (0,1) -- (1,2) -- (2,3) -- (3,4) -- (4,5) -- (5,6) -- (6,7) -- (8,1);

\draw [thick] (6,-.1) node[below]{\Huge $q_r$} -- (6,0.1);
\draw[-,dotted,color=gray] (6,0) -- (6,7);

\node [circle,fill=black] at (2,3) {};
\draw [thick] (2,-.1) node[below]{\Huge $\theta_R(S_{t'})$} -- (2,0.1);
\draw[-,dotted,color=gray] (2,0) -- (2,3);

\draw[->,thick] (2,3) -- (7,4);
\draw [thick] (7,-.1) node[below]{\Huge $\theta_R(S_{t})$} -- (7,0.1);
\draw[-,dotted,color=gray] (7,0) -- (7,4);

\end{tikzpicture}
\caption{Virtual peak of a new agent.}\label{fig:new}
\end{subfigure}%
\caption{The red ratio of a coalition does not exceed the virtual peak of a default agent, 
but exceeds the virtual peak of a new agent.}
\end{figure}
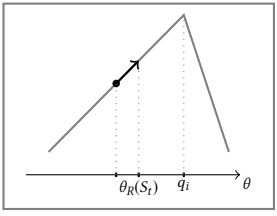
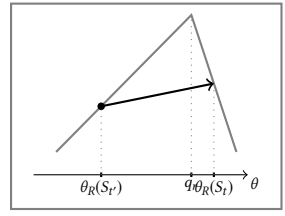

\noindent {\bf Deviations of red agents\ } 
We will now show that coalitions $S_1,\ldots,S_k$ do not admit an IS-deviation by red agents. 
To this end, we first observe that the red ratio of a coalition to which a new agent belongs 
exceeds the virtual peak of the new agent; see Figure~\eqref{fig:new} for an illustration. Due to 
single-peakedness, this means that new agents do not accept further deviations of red agents.

\begin{lemma}\label{lem:red:while}
Let $S$ be the set of agents that belong to $S_{t}$ with $t>0$ 
after the {\bf while}-loop of Step \ref{step:red:while} of Algorithm~\ref{alg:sub:IS}, 
and let $r_i$ be the last agent in $S$ to join $S_{t}$. 
Then no red agent $r_j$ with $j>i$ can join $S$ without exceeding the minimum peak, i.e., 
$
\theta_R(S \cup \{r_j\}) > \min \{q_{b_{t}},q_{r_j}\}
$.
\end{lemma}
\begin{proof}
Take any $r_j$ with $j>i$. If $\theta_R(S \cup \{r_j\}) \leq \min \{q_{b_{t}},q_{r_j}\}$, then this 
would mean that $\theta_R(S \cup \{r_{i+1}\}) \leq \min \{q_{b_{t}},q_{r_{i+1}}\}$, 
since $q_{r_{i+1}} \geq q_{r_j}$. Hence the algorithm would have added $r_{i+1}$ to $S$, a contradiction.
\end{proof}

\begin{lemma}\label{lem:new}
For all $S_{t} \in \pi$ with $t>0$ and each new agent $r\in S_t$:
\begin{enumerate}
\item the red ratio of $S_t$ exceeds $r$'s virtual peak, i.e., $\theta_R(S_{t}) > q_{r}$; 
\item $r$ is the unique new agent in $S_{t}$; 
\item $r$ weakly prefers $S_t$ to $\{r\}$ and to a mixed pair.
\end{enumerate}
\end{lemma}
\begin{proof}
Let $r$ be the first new agent who joined $S_{t}$ at Step \ref{step:leave} or Step 
\ref{step:outside}. When $r$ joined $S_t$, his virtual peak $q_r$ 
was at most that of any other red agent in $S_t$. 
By Lemma \ref{lem:red:while}, he cannot join $S_t$ without exceeding the virtual peak $q_r$, 
i.e., $\theta_R(D_{t}\cup \{r\}) < q_r$. Thus, by single-peakedness, he does not accept any 
further deviations by red agents. Hence, $r$ is the unique new agent in $S_{t}$. To see that the 
third statement holds, note that if $r$ joined $S_{t}$ at Step \ref{step:outside}, then $r$ prefers 
$\theta_R(S_{t})$ to $\frac{1}{2}$ by the {\bf if}-condition of Step \ref{step:outside:if}. Also, $r$ 
chooses his favorite coalition among the coalitions to which he can deviate, and thus $r$ prefers 
$\theta_R(S_{t})$ to $\theta_R(S_{0}\cup\{r\})=1$. Now suppose that $r$ joined $S_{t}$ from $S_{t'}$ 
with $t'>t$ at Step \ref{step:leave}. At that point, the ratio $\theta'$ of red agents in $S_{t'}$ was 
at least half, and agent $r$ weakly prefers $S_{t'}$ to a mixed pair 
since $\frac{1}{2} \leq \theta' \leq q_r$. 
Hence, since $r$ prefers $\theta_R(S_{t})$ to $\theta'$, by transitivity $r$ prefers $\theta_R(S_{t})$ to 
$\frac{1}{2}$. Also, by the {\bf if}-condition in Step \ref{step:leave:if}, 
$r$ prefers $S_{t}$ to being in a singleton coalition.
\end{proof}

We also observe that the red ratio of a higher-indexed coalition 
is smaller than or equal to that of a lower-indexed coalition. 

\begin{lemma}\label{lem1:ratio}
Let $t' \geq t>0$ and let $S$ be the set of agents in $S_t$ just after the {\bf while}-loop of 
Step~\ref{step:red:while}. Then 
$\theta_R(D_{t'}) \leq \theta_R(S)$; in particular, $\theta_R(D_{t'}) \leq \theta_R(D_{t})$.
\end{lemma}
\begin{proof}
Assume for the sake of contradiction that $\theta_R(D_{t'})>\theta_R(S)$. 
Since $S$ and $D_{t'}$ contain exactly one blue agent, 
this means that $D_{t'}$ contains more red agents than $S$ does, i.e., 
$|S \cap R|+1 \leq |D_{t'} \cap R|$. Thus, even if we add a red agent to $S$,  
its red ratio does not exceed $\theta_R(D_{t'})$, i.e., for each 
$r \in R \setminus S$ we have
$\theta_R(S \cup \{r\}) \leq \theta_R(D_{t'})$.
Now recall that by Lemma \ref{lem:default}, the red ratio of $S_{t'}$ 
does not exceed the virtual peak of each default agent, which means that
\[
\theta_R(D_{t'}) \leq \theta_R(S_{t'}) \leq  \min_{x \in D_{t'}} q_x.
\]
Combining these observations, for every $r \in R\setminus S$ we have 
\begin{equation}\label{eq}
\theta_R(S \cup \{r\}) \leq \min_{x \in D_{t'}} q_x.
\end{equation}
Moreover, since $t=t'$, or $S_{t'}$ has been created after $S_{t}$, there is a red agent 
$r_j \in D_{t'} \setminus S$ with $q_{r_j} \leq q_{r}$ 
for every $r \in S \cap R$ and $q_{b'} \leq q_{b}$, where 
$b'$ and $b$ are the unique blue agents in $D_{t'}$ and $S$, respectively. 
Combining this with inequality~\eqref{eq} yields
\[
\theta_R(S \cup \{r_j\}) \leq \min_{x \in S \cup \{r_j\}} q_x,
\]
which contradicts Lemma \ref{lem:red:while}. 
\end{proof}

\begin{lemma}\label{lem2:ratio}
Let $t'>t>0$ with $\theta_R(D_{t'})= \theta_R(D_{t})$. 
Then some agent in $D_{t'}$ does not accept a deviation of a red agent to $S_{t'}$.
\end{lemma}
\begin{proof}
Since $\theta_R(D_{t'})=\theta_R(D_{t})$, both $D_{t}$ and $D_{t'}$ 
contain exactly one blue agent as well as the same number of red agents. 
By the description of the algorithm, the minimum virtual peak of agents in $D_{t'}$ 
is smaller than that of agents in $D_{t}$, i.e., 
\begin{align}
\min_{x \in D_{t'}}q_x \leq \min_{x \in D_{t}}q_x \label{eq2}.
\end{align}
Further, by Lemma \ref{lem:red:while}, no red agent in $D_{t'}$ can join $D_{t}$ 
without exceeding the minimum virtual peak, which implies that for every $r \in D_{t'}$
we have
\[
\min_{x \in D_{t} \cup \{r\}}q_x < \theta_R(D_t \cup \{r\}) = \frac{|D_{t'}\cap R|+1}{|D_{t'}|+1}. 
\]
Combining this with the inequality~\eqref{eq2} implies 
\[
\min_{x \in D_{t'}}q_x < \frac{|D_{t'}\cap R|+1}{|D_{t'}|+1}.
\]
Thus, by single-peakedness, there is an agent in $D_{t'}$ who does not accept 
a deviation of a red agent to $D_{t'}$. 
\end{proof}

We are now ready to show that no coalition in $\pi$ admits a deviation of a red agent.

\begin{lemma}\label{thm:mix1}
No agent $r \in S_{t} \cap R$ has an IS-deviation to another coalition $S_{t'}$ with $0<t' < t$. 
\end{lemma}
\begin{proof}
Suppose towards a contradiction that there is such 
an agent $r \in S_{t} \cap R$ and a coalition $S_{t'}$ with $0<t'< t$. 
Let $t'$ be the largest index of coalition $S_{t'}$ to which $r$ has an IS-deviation.
Observe that if $S_{t'}$ contains a new agent, then the new agent does not accept a deviation of agent 
$r$ by Lemma~\ref{lem:new}; thus, $S_{t'}$ has no new agent and $S_{t'}=D_{t'}$.
If $r$ is a default agent in $S_{t}$, then $r$ weakly prefers $S_{t}$ to $D_{t}$ and hence $r$ strictly 
prefers $D_{t'} \cup \{r\}$ to $D_{t}$ by transitivity. Thus, agent $r$ could have deviated to $S_{t'}$ 
from $S_{t}$ at Step \ref{step:leave}, a contradiction. If $r$ is a new agent in $S_{t}$, then $r$ 
joined $S_{t}$ in Step \ref{step:leave} or Step \ref{step:outside}, but this means that $r$ could have 
deviated to $S_{t'}$ instead of $S_{t}$, 
as $r$ strictly prefers $D_{t'} \cup \{r\}$ to $D_{t} \cup \{r\}$, a contradiction.
\end{proof}

\begin{lemma}\label{thm:mix2}
No agent $r \in S_{t} \cap R$ has an IS-deviation to another coalition $S_{t'}$ with $t' > t$. 
\end{lemma}
\begin{proof}
Suppose towards a contradiction that there is such an agent $r \in S_{t} \cap R$ and 
a coalition $S_{t'}$. Let $t'>t$ be the 
smallest index of coalition $S_{t'}$ to which $r$ has an IS-deviation. Again, if $S_{t'}$ contains a 
new agent, then $r$ cannot deviate to $S_{t'}$; thus, $S_{t'}=D_{t'}$, and 
$\theta_R(S_{t'})=\theta_R(D_{t'}) \leq \theta_R(D_{t}) \leq \theta_R(S_{t})$ by Lemma~\ref{lem1:ratio}.

First, consider the case where agent $r$ is a default agent in $S_{t}$ and $t>0$. 
Then, by Lemma \ref{lem:default} and by the fact that $\theta_R(S_{t'}) \leq \theta_R(S_{t})$, we have
$\theta_R(S_{t'}) \leq \theta_R(S_{t}) \leq q_r$. 
Since $r$'s preferences are single-peaked, 
$r$ has an incentive to join $S_{t'}$ only if $\theta_R(S_{t'})=\theta_R(S_{t})$; 
however, by Lemma \ref{lem2:ratio} this implies that there is an agent in $S_{t'}$ 
who is not willing to accept $r$'s deviation, a contradiction. 

Second, consider the case where agent $r$ joined $S_{t}$ at Step \ref{step:leave}. Let 
$S_{\ell}$ be the coalition to which the red agent $r$ belonged before joining $S_{t}$, 
and let $S$ be the set of agents in $S_{\ell}$ just before $r$ deviated from $S_{\ell}$ to $S_t$ at Step 
\ref{step:leave}. Note that $S_{\ell}$ is a coalition that has been created before $S_{t'}$ emerged, 
i.e., $\ell \leq t'$, since otherwise $r$ would have deviated to $S_{t'}$ instead of $S_{t}$ at Step 
\ref{step:leave}. By Lemma \ref{lem1:ratio}, $\theta_R(D_{t'})$ is at most $\theta_R(S)$
and $\theta_R(S)\le q_r$ by the description of the algorithm. Now we have
$\theta_R(D_{t'}) \leq \theta_R(S) \leq q_r < \theta_R(S_{t})$.
 
Recall that $r$ strictly prefers $S_{t}$ to $S$; thus, by single-peakedness, 
$r$ has an incentive to deviate from $S_{t}$ to $S_{t'}$ only if 
$\theta_R(D_{t'}) = \theta_R(S)$. Since $\theta_R(D_{\ell}) \leq \theta_R(S)$ 
by Lemma \ref{lem1:ratio}, this means that $\theta_R(D_{t'})= \theta_R(D_{\ell})$. 
Also, if $\ell < t'$, some agent in $D_{t'}$ does not accept the deviation of $r$ 
by Lemma \ref{lem2:ratio}, and hence $\ell = t'$. Thus, agent $r$ deviated from 
$S_{t'}$ to $S_t$ and later the coalition $S_{t'}$ regained the same number of red agents as before. 
Now, if $\theta_R(S)=q_r$, then $r$ would not have left the coalition $S$ at Step \ref{step:leave}; 
hence $\theta_R(S)<q_r$, which implies 
$\theta_R(D_{t'} \cup \{r\}) \leq q_r$. 
Recall that by Lemma \ref{lem:default}, the red ratio of $D_{t'}$ 
is at most the minimum virtual peak, i.e., 
$\theta_R(D_{t'}) \leq \min_{x \in D_{t'}} q_x$. 
If adding $r$ to $D_{t'}$ exceeds the minimum, some agent would not 
accept a deviation of $r$; thus, we have 
\begin{equation}\label{eq:red}
\theta_R(D_{t'}\cup \{r\}) \leq \min_{x \in D_{t'} \cup \{r\}} q_x.
\end{equation}
Now, since $\theta_R(S \setminus \{r\})<\theta(D_{t'})$, there is at least one red agent 
$r_j \in D_{t'}$ who does not belong to $S$. By \eqref{eq:red}, 
\[
\theta(S \cup \{r_j\}) = \theta(D_{t'}\cup \{r\}) \leq \min_{x \in D_{t'} \cup \{r\}} q_x \leq \min_{x \in S \cup \{r_j\}} q_x,
\]
contradicting Lemma \ref{lem:red:while}. 

Finally, consider the case where $r$ joined $S_{t}$ at Step \ref{step:outside}. Since 
$S_{t'}$ does not contain any new agents, agent $r$ could have joined $S_{t'}$ instead of $S_t$ at Step 
\ref{step:outside}, a contradiction.
\end{proof}


\subsection{Algorithm for individual stability}
We will now construct an individually stable outcome
using the algorithm described in Section~\ref{sec:mixed} as a subroutine. 
Suppose that we are given a diversity game $(R,B,(\succ_i)_{i \in R\cup B})$. 
For each $i \in R \cup B$, we denote by $\succ^B_i$ the preference over 
the ratios of the blue agents in each coalition: $\theta_1 \succ^B_i \theta_2$ if and only if 
$(1-\theta_1) \succ_i (1-\theta_2)$. 

\begin{algorithm}
\SetKw{And}{and}
\SetKw{None}{None}
\SetKwInOut{Input}{input} 
\SetKwInOut{Output}{output}
\caption{Algorithm for an individually stable outcome}\label{alg:IS}
\Input{A single-peaked diversity game $(R,B,(\succ_i)_{i \in R\cup B})$}
\Output{$\pi$}
	make mixed coalitions with red majority, i.e., set $\pi_R=\HALF(R,B,(\succ_i)_{i \in R\cup B})$ \label{step:red}\;
	let $R_{\remain} \leftarrow R \setminus \bigcup_{S \in \pi_R}S$ and $B_{\remain}\leftarrow B \setminus \bigcup_{S \in \pi_R}S$\;
	make mixed coalitions with blue majority, i.e., set $\pi_B=\HALF(B_{\remain},R_{\remain},(\succ^B_i)_{i \in R_{\remain} \cup B_{\remain}})$\;
	set $R_{\remain} \leftarrow R_{\remain} \setminus  \bigcup_{S \in \pi_B}S$, and $B_{\remain}\leftarrow B_{\remain} \setminus  \bigcup_{S \in \pi_B}S$\;
	make mixed pairs from remaining agents in $R_{\remain}$ and $B_{\remain}$ who prefer a coalition of ratio $\frac{1}{2}$ to his or her own singleton; add them to $\pi_{\pair}$ \label{step:pair}\;
	set $R_{\remain} \leftarrow R_{\remain} \setminus  \bigcup_{S \in \pi_{\pair}}S$, and $B_{\remain} \leftarrow B_{\remain} \setminus \bigcup_{S \in \pi_{\pair}}S$\;
	{\tt //  let the remaining agents deviate to mixed coalitions as long as they prefer the deviating coalition to his or her singleton}\;
	\ForEach{$A \in \{R,B\}$}{
	\While{there is a agent $i \in A_{\remain}$ and a mixed coalition $S \in \pi_A$ such that $S \cup \{i\}$ is individually rational and all agents in $S$ accept a deviation of $i$ to $S$\label{step:lastdeviation}}{
	choose such pair $i$ and $S$ where $\theta_A(S \cup \{i\})$ is $i$'s most preferred ratio among the coalitions $S$ satisfying the above\;
	$\pi_A \leftarrow \pi_A \setminus \{S\} \cup \{S \cup \{i\}\}$, and $A_{\remain} \leftarrow A_{\remain} \setminus \{i\}$\label{step:outside2}\;
	}
	}
	put all the remaining agents in $R_{\remain}$ and $B_{\remain}$ into singletons, and add them to $\pi_{\single}$\;
	return $\pi=\pi_R \cup \pi_B  \cup \pi_{\pair} \cup \pi_{\single}$\;
\end{algorithm}

Now let $\pi_R$, $\pi_B$, $\pi_{\pair}$, and $\pi_{\single}$ denote the partitions computed 
by Algorithm \ref{alg:IS}, and let $\pi=\pi_R \cup \pi_B \cup \pi_{\pair} \cup \pi_{\single}$.  
Arguing as in the previous section, we can establish the following properties.

\begin{lemma}\label{lem}
For each $A \in \{R,B\}$ and each coalition $S\in\pi_A$ it holds that:
\begin{enumerate}
\item[(a)] $S$ is individually rational; 
\item[(b)] if $i$ joined $S$ at Step \ref{step:outside2}, 
        then $i$ does not accept a further deviation of an agent in $A$ to $S$; 
\item[(c)] there is an agent who does not admit a deviation 
             of an agent in $N \setminus A$;
\item[(d)] for each agent $i \in S$ who joined $S$ 
     before Step~\ref{step:outside2} of Algorithm \ref{alg:IS}, $i$ weakly prefers $S$ to a mixed pair. 
\end{enumerate}
\end{lemma}
\begin{proof}
Without loss of generality suppose that $A$ is the set of red agents. Claim~(a) 
(individual rationality) holds due to Lemmas~\ref{lem:default} and~\ref{lem:new} 
and the {\bf while}-condition in Step \ref{step:lastdeviation} of Algorithm~\ref{alg:IS}.

To show claim~(b), suppose that a red agent $r$ is the first agent who joined 
$S\in \pi_R$ at Step \ref{step:outside2}. By single-peakedness, the claim holds when 
$p_r \leq \frac{1}{2}$; so suppose $p_r > \frac{1}{2}$. Since $r$ did not belong 
to any of the coalitions in $\pi_R$ before, his virtual peak $q_r$ is 
at most that of any other red agent in $S$. Since all agents in $S \setminus \{r\}$ 
accept $r$ at Step \ref{step:lastdeviation} of Algorithm \ref{alg:IS}, 
$S$ did not contain a new agent by Lemma \ref{lem:new}, so the red ratio of 
$S \setminus \{r\}$ is at most the minimum virtual peak of agents 
in $S \setminus \{r\}$ by Lemma \ref{lem:default} and this remains true after accepting $r$. 
However, by Lemma \ref{lem:red:while}, agent $r$ cannot join $S$ without exceeding 
either the virtual peak of the unique blue agent or the virtual peak $q_r$; 
thus, $\theta_R(S) < q_r$. By single-peakedness, we conclude that $r$ 
does not accept any further deviation of red agents. 

To prove claim~(c), we can use the argument in the proof of Lemma \ref{lem:blue}. 
To establish claim~(d), we recall that each agent $i \in S$ 
weakly prefers $S$ to a mixed pair by Lemmas~\ref{lem:default} and~\ref{lem:new}; 
by transitivity, this still holds after accepting the deviation of red agents at 
Step \ref{step:outside2} of Algorithm \ref{alg:IS}. 
\end{proof}

\begin{lemma}\label{lem:pair}
Mixed pairs created at Step \ref{step:pair} of Algorithm \ref{alg:IS} do not admit an IS-deviation.
\end{lemma}
\begin{proof}
Let $R'$ and $B'$ be the set of red and blue agents in $R_{\remain}$ and $B_{\remain}$ 
just before Step \ref{step:pair} of Algorithm \ref{alg:IS}, respectively.
Note that both $R'$ and $B'$ are non-empty only if red agents 
(respectively, blue agents) are left when mixed coalitions with red majority have been formed, 
and blue agents (respectively, red agents) are left when mixed coalitions with blue majority 
have been formed. So, the red agents in $R'$ and the blue agents in $B'$ 
have opposite peaks, i.e., we have either
\begin{itemize}
\item $p_r \geq \frac{1}{2}$ for all $r \in R'$ and $p_b \leq \frac{1}{2}$ for all $b \in B'$; or 
\item $p_r \leq \frac{1}{2}$ for all $r \in R'$ and $p_b \geq \frac{1}{2}$ for all $b \in B'$.
\end{itemize}
Thus, for each agent who would like to join, 
one of the agents in the pair would not accept such a deviation. 
\end{proof}

\begin{lemma}
The partition $\pi$ is individually stable.
\end{lemma}
\begin{proof}
We first observe that $\pi$ is individually rational. 
Indeed, all singletons in $\pi$ are individually rational. 
Also, all mixed pairs are individually rational, 
since they prefer a coalition of ratio $\frac{1}{2}$ to being in a singleton coalition. 
All coalitions in $\pi_R$ and $\pi_B$ are individually rational by Lemma \ref{lem}(a).

We will now show that $\pi$ satisfies individual stability. 
Take any red agent $r \in R$. We note that by individual rationality $r$ has no incentive to deviate to 
red-only coalitions; also, by Lemma \ref{lem:pair}, $r$ cannot deviate to a coalition in 
$\pi_{\pair}$. Further, by Lemma \ref{lem}(c), agent $r$ has no IS-deviation to a coalition in 
$\pi_B$. Thus, it remains to check whether $r$ has an IS-deviation to a coalition in $\pi_R$, or a blue 
singleton in $\pi_{\single}$. Consider the following cases:

(1) Agent $r$ joined $S \in \pi_R$ before Step \ref{step:outside2} of Algorithm \ref{alg:IS}. \\
By Lemmas~\ref{lem:default} and~\ref{lem:new}, agent $r$ weakly prefers $S$ to a mixed pair before 
accepting the deviation of red agents at Step \ref{step:outside2} of Algorithm \ref{alg:IS}, and by 
transitivity, this remains true after Step \ref{step:outside2} of Algorithm \ref{alg:IS}, meaning that 
$r$ does not have an incentive to join a blue singleton. Also, $r$ has no IS-deviation to a coalition 
in $\pi_R$ by Lemmata \ref{thm:mix1} and \ref{thm:mix2}.

(2) Agent $r$ joined $S \in \pi_R$ at Step \ref{step:outside2} of Algorithm \ref{alg:IS}. \\
By individual rationality, agent $r$ weakly prefers his coalition to forming his own singleton. Hence, if $r$ 
has an IS-deviation to a coalition of a single blue agent in $\pi_{\single}$, then it means that both 
agent strictly prefer a mixed pair to their own singleton coalitions and hence they would have formed a 
pair at Step \ref{step:pair}, a contradiction. Also, if agent $r$ has an IS-deviation to some coalition 
$T \in \pi_R$, then it means that no red agent joined $T$ at Step \ref{step:outside2} of Algorithm 
\ref{alg:IS} by Lemma \ref{lem}(c), and thus agent $r$ would have joined $T$ instead of $S$, a 
contradiction.

(3) Agent $r$ belongs to a coalition $S \in \pi_B$. \\
By construction of the algorithm, we have $p_r \leq \frac{1}{2}$, and 
agent $r$ weakly prefers $S$ to a mixed pair before accepting the deviation of blue agents at Step 
\ref{step:outside2} of Algorithm \ref{alg:IS}. By transitivity, this remains true after Step 
\ref{step:outside2} of Algorithm \ref{alg:IS}, which means that $r$ prefers his coalition 
to a coalition with blue ratio at most $\frac{1}{2}$. 
Thus, $r$ has no incentive to deviate to a coalition in $\pi_R$, or to a blue-only coalition.

(4) Agent $r$ belongs to a coalition in $\pi_{\pair}$. \\
Clearly, agent $r$ has no incentive to deviate to singleton coalitions of blue agents. 
Also, if $r$ has an IS-deviation to some coalition $T \in \pi_R$, $r$ would have joined $T$ at Step 
\ref{step:outside} of Algorithm \ref{alg:sub:IS}, a contradiction.

(5) Agent $r$ belongs to a coalition in $\pi_{\single}$. \\
Again, if $r$ has an IS-deviation to some coalition $T \in \pi_R$, $r$ could have joined $T$ at Step 
\ref{step:outside2} of Algorithm \ref{alg:IS}, a contradiction. Also, if $r$ has an IS-deviation to 
some coalition in $\pi_{\single}$ that consists of a single blue agent $b$, then both $r$ and $b$ could 
have formed a pair at Step \ref{step:pair}, a contradiction.


A symmetric argument applies to blue agents' deviations, and hence no blue agent has an IS-deviation to 
other coalitions.
\end{proof}

\section{Conclusion}
We have initiated the formal study of coalition formation games with varying degree of 
homophily and heterophily. Our results suggest several directions for future work. 

First, while we have argued that Nash stable outcomes may fail to exist, the complexity
of deciding whether a given diversity game admits a Nash stable outcome remains unknown.
Also, we have obtained an existence result for individual stability
under the assumption that agent's preferences are single-peaked, but it is not clear
if the single-peakedness assumption is necessary; in fact, we do not have an example
of a diversity game with no individually stable outcome. In a similar vein, it would be 
desirable to identify further special classes of diversity games that admit core stable
outcomes or at least efficient algorithms for deciding whether the core is non-empty.

More broadly, it would be interesting to extend our model to more than two agent types.
Another possible extension is to consider the setting where agents are located
on a social network, and each agent's preference over coalitions 
is determined by the fraction of her acquaintances in these coalitions;
this model would capture both the setting considered in our work 
and fractional hedonic games.

\bibliographystyle{ACM-Reference-Format}

\end{document}